\documentclass[a4paper]{amsart}
\usepackage{amsmath,amssymb}
\usepackage{graphicx}

\usepackage{hyperref}
\usepackage[round]{natbib}
\usepackage[font=footnotesize]{caption}

\theoremstyle{plain}
\newtheorem{thm}{Theorem}[section]
\newtheorem{cor}[thm]{Corollary}
\newtheorem{lem}[thm]{Lemma}

\theoremstyle{definition}
\newtheorem{Def}[thm]{Definition}
\newtheorem{example}[thm]{Example}
\theoremstyle{remark}
\newtheorem{Rem}{Remark}[thm]

\numberwithin{equation}{section}

\DeclareMathOperator{\ext}{ext}
\DeclareMathOperator{\conv}{conv}

\newcommand{\rcvx}[1]{{#1}^{\mathbb{Q}}}

\newcommand{\pu}{\mathrm{pu}}
\newcommand{\nc}{\mathrm{nc}}
\newcommand{\co}{\mathrm{co}}

\begin{document}


\title{The Payoff Region of a Strategic Game and Its Extreme Points}
\author{Yu-Sung Tu}
\address{Institute of Economics\\
         Academia Sinica\\
         Taipei 11529, Taiwan}
\email[Yu-Sung Tu]{yusungtu@gate.sinica.edu.tw}
\author{Wei-Torng Juang}
\email[Wei-Torng Juang]{wjuang@econ.sinica.edu.tw}


\begin{abstract}
The range of a payoff function for an $n$-player finite strategic game is investigated using a novel approach, the notion of extreme points of a non-convex set.
The shape of a noncooperative payoff region can be estimated using extreme points and supporting hyperplanes of the cooperative payoff region.
A basic structural characteristic of a noncooperative payoff region is that any of its subregions must be non-strictly convex if the subregion contains a relative neighborhood of a point on its boundary.
Besides, applying the properties of extreme points of a noncooperative payoff region is a simple and effective way to prove some results about Pareto efficiency and social efficiency in game theory.
\end{abstract}

\maketitle


\section{Introduction}


In this paper, we attempt to explore the fundamental properties of noncooperative payoff regions and their potential applications. A payoff region of an $n$-player finite strategic game can occur under different hypotheses about the players' strategic behavior. The cooperative payoff region achievable under correlated strategies is the convex hull of the payoff vectors generated by pure-strategy profiles; the noncooperative payoff region achievable under mixed-strategy profiles is a subset of the convex polytope just described.
Compared to the cooperative payoff region, the shape of the noncooperative payoff region is usually too complex to be characterized simply.

It is very common that a noncooperative payoff region of a two-player finite strategic game could look like a sharp boomerang,
but it would never look like a ice cream cone (see, e.g., \citet{kBin:pr} and \citet{eBar:gt}).
Here this basic characteristic will be proved in a mathematically rigorous manner. We introduce the concept of extreme points of a non-convex set, and apply it to a noncooperative payoff region. The key theorem is that for an $n$-player finite strategic game, any extreme point of the cooperative payoff region is an extreme point of the noncooperative payoff region, and \emph{all} these extreme points can be achieved as the payoff profiles on which the players choose pure strategies. This allows us to deduce directly that any subset of an $n$-dimensional noncooperative payoff region cannot be \emph{strictly convex} if it contains a relative neighborhood of a boundary point of the payoff region.

On the other hand, it is easy to see that the cooperative and noncooperative payoff regions generated from an $n$-player finite strategic game have the same supporting hyperplanes. Moreover, every supporting hyperplane to a noncooperative payoff region in $\mathbb{R}^n$ must contain at least one extreme point of this payoff region, no matter what the shape of this patoff region is. These results will be presented in Section~\ref{sec:PayReg}.

An extreme point of a noncooperative payoff region can be achieved through the play of a pure-strategy profile, but not vice versa.
We naturally compare it to the payoff profile which cannot be achieved with a non-degenerate mixed-strategy profile. In Section~\ref{sec:examples}, we give some examples to show that in fact there is no necessary relationship between them. Finally, we present in Section~\ref{sec:App} some applications regarding Pareto efficiency and social efficiency in the theory of games. These will demonstrate that the approach using the properties of extreme points of noncooperative payoff regions is a simple and effective method of proving theorems.


\section{Preliminaries and Basic Properties}
\label{sec:Basic}


We shall consider finite strategic games. Let $N = \{1, \dots, n\}$ be the set of players.
For $i\in N$, the nonempty finite set $A_i$ is the set of \emph{pure strategies} available to player~$i$.
The set of pure-strategy profiles is the Cartesian product $\prod_{i\in N}A_i$ of the players' pure-strategy sets.
A \emph{mixed strategy} $\sigma_i$ of player~$i$ is a probability distribution over $A_i$, and let $\Delta(A_i)$ denote the set of mixed strategies of player~$i$.
The Cartesian product $\prod_{i\in N}\Delta(A_i)$ is the set of all mixed-strategy profiles.
We define a \emph{correlated strategy} $\varphi$ for $n$ matched players to be a probability distribution over $\prod_{i\in N}A_i$, and denote by $\Delta(\prod_{i\in N}A_i)$ the set of correlated strategies. Then every mixed-strategy profile $\sigma\in \prod_{i\in N}\Delta(A_i)$ can induce a correlated strategy $\varphi_{\sigma}$ in the way: $\varphi_{\sigma}(a) = \prod_{i\in N} \sigma_i(a_i)$ for every $a\in \prod_{i\in N}A_i$, where $\sigma_i(a_i)$ is the probability assigned by $\sigma_i$ to $a_i$. We call such a correlated strategy the \emph{induced correlated strategy of} $\sigma$.

For $i\in N$, let $u_i\colon \prod_{i\in N}A_i\to \mathbb{R}$ be the payoff function of player~$i$. Every payoff function $u_i$ can be extended to the set $\Delta(\prod_{i\in N}A_i)$ by taking the expected values over $\prod_{i\in N}A_i$. We can also extend its domain to the set $\prod_{i\in N}\Delta(A_i)$ in such a way that the payoff value of $u_i$ at the mixed-strategy profile $\sigma\in \prod_{i\in N}\Delta(A_i)$ is
\begin{equation}\label{eq:20170303}
u_i(\sigma) = u_i(\varphi_{\sigma}) = \sum_{a\in \prod_{i\in N}A_i} \varphi_{\sigma}(a) u_i(a),
\end{equation}
where $\varphi_{\sigma}$ is the induced correlated strategy of $\sigma$.
\footnote{The following is an equivalent definition for $u_i$ extended to the set of mixed-strategy profiles:
for any $\sigma\in \prod_{i\in N}\Delta(A_i)$, let
\[
u_i(\sigma) = \sum_{a_1\in A_1}\dots \sum_{a_n\in A_n} \sigma_1(a_1) \cdots \sigma_n(a_n) u_i(a_1, \dots, a_n).
\]
}
Define the vector-valued payoff function $u\colon \prod_{i\in N}A_i\to \mathbb{R}^n$ by $u(a) = (u_1(a),\dots,u_n(a))$. Then $u$ can be extended to the sets $\Delta(\prod_{i\in N}A_i)$ and $\prod_{i\in N}\Delta(A_i)$ through the $u_i$.

For $i\in N$ and $a_i\in A_i$, let $\delta_{a_i}$ denote the degenerate probability (Dirac measure) concentrated at $a_i$. It is clear that a pure-strategy profile $a = (a_1, \dots, a_n)$ will correspond to its mixed-strategy profile $\sigma_a = (\delta_{a_1}, \dots, \delta_{a_n})$, and we have
\[
u(\sigma_a) = u(a).
\]
So we can embed the set of pure-strategy profiles into the set of mixed-strategy profiles.
In fact, such a relationship also exists between the set of mixed-strategy profiles and the set of correlated strategies.
The following lemma shows that distinct mixed-strategy profiles correspond to distinct correlated strategies.

\begin{lem}
In a finite strategic game, the set of all mixed-strategy profiles and the set of all induced correlated strategies are in one-to-one correspondence.
\end{lem}

\begin{proof}
Let $|A_i| = m_i$ for each $i\in N$. Define $f\colon \prod_{i\in N}\Delta(A_i)\to \Delta(\prod_{i\in N}A_i)$ by $f(p) = \varphi_p$, where $p = (p_1, \dots, p_n)$, $p_i = (p_i^1, \dots, p_i^{m_i})$ for all $i\in N$, and $\varphi_p$ is the induced correlated strategy of $p$. It is clear that the range of $f$ is the set of all induced correlated strategies. To see that $f$ is injective, suppose that $f(p) = f(q)$. Then $\prod_{i\in N}p_i^{r_i} = \prod_{i\in N}q_i^{r_i}$ for any $(r_i)_{i\in N}$ with $1\leq r_i\leq m_i$. The fact that
\[
\sum_{r_j=1}^{m_j}p_j^{r_j}\prod_{i\in N\setminus \{j\}}p_i^{r_i} = \sum_{r_j=1}^{m_j}q_j^{r_j}\prod_{i\in N\setminus \{j\}}q_i^{r_i}
\]
for each $j\in N$ and for each $(r_i)_{i\in N\setminus \{j\}}$ with $1\leq r_i\leq m_i$ implies that the equality $\prod_{i\in N\setminus \{j\}}p_i^{r_i} = \prod_{i\in N\setminus \{j\}}q_i^{r_i}$ holds for any $j\in N$ and for any $(r_i)_{i\in N\setminus \{j\}}$ with $1\leq r_i\leq m_i$.

Again, we use the fact that
\[
\sum_{r_k=1}^{m_k}p_k^{r_k}\prod_{i\in N\setminus \{j, k\}}p_i^{r_i} = \sum_{r_k=1}^{m_k}q_k^{r_k}\prod_{i\in N\setminus \{j, k\}}q_i^{r_i}
\]
for each $j$, $k\in N$ and for each $(r_i)_{i\in N\setminus \{j, k\}}$ with $1\leq r_i\leq m_i$. This fact implies the equality $\prod_{i\in N\setminus \{j, k\}}p_i^{r_i} = \prod_{i\in N\setminus \{j, k\}}q_i^{r_i}$ for any $j$, $k\in N$ and for any $(r_i)_{i\in N\setminus \{j, k\}}$ with $1\leq r_i\leq m_i$. Repeat this process. The pattern is clear, and eventually we reach the conclusion that $p_i^{r_i} = q_i^{r_i}$ for each $i\in N$ and for each $r_i\in \{1, \dots, m_i\}$.
\end{proof}

Therefore, applying the payoff relation~\eqref{eq:20170303}, we can embed the set of mixed-strategy profiles into the set of correlated strategies.

\begin{Rem}\label{rem:20160418}
In every finite strategic game, the set $\prod_{i\in N}A_i$ can be considered as embedded in the set $\prod_{i\in N}\Delta(A_i)$; the set $\prod_{i\in N}\Delta(A_i)$ can be considered as embedded in the set $\Delta(\prod_{i\in N}A_i)$.
\end{Rem}


\section{Noncooperative Payoff Regions}
\label{sec:PayReg}


A payoff region of a game can occur under different hypotheses about what players will be able to do. It is well known that the payoff region may be very far from convex when the players choose their strategies independently. It would be helpful if we could obtain further information about the shape of a noncooperative payoff region.

In this section, the concept of extreme points of a non-convex set is introduced, and we apply it to a noncooperative payoff region, together with supporting hyperplanes.
This approach can give us valuable insight into understanding the general shape of a non-convex payoff region.

\begin{Def}
Let $(N, (A_i)_{i\in N}, (u_i)_{i\in N})$ be a finite strategic game. The three ranges $u(\prod_{i\in N}A_i)$, $u(\prod_{i\in N}\Delta(A_i))$, and $u(\Delta(\prod_{i\in N}A_i))$ are denoted respectively by $S_{\pu}$, $S_{\nc}$, and $S_{\co}$, and are called respectively the \emph{pure-payoff region}, the \emph{noncooperative payoff region}, and the \emph{cooperative payoff region}.
\end{Def}

From Section~\ref{sec:Basic}, we know that $S_{\pu}\subseteq S_{\nc}\subseteq S_{\co}$. The pure-payoff region $S_{\pu}$ is a finite subset of $\mathbb{R}^n$, and the convex hull of $S_{\pu}$, denoted by $\conv(S_{\pu})$, is just the cooperative payoff region $S_{\co}$, a convex polytope in $\mathbb{R}^n$. The noncooperative payoff region $S_{\nc}$ is a colsed subset of $S_{\co}$, and it is also a set of generators of $S_{\co}$, that is,
\begin{equation}\label{eq:20170309}
S_{\co} = \conv(S_{\pu}) = \conv(S_{\nc}).
\end{equation}
In this paper, we shall provide more of the basic features of a noncooperative payoff region.

First, we show that the noncooperative payoff region of an $n$-player finite strategic game is a closed, bounded, and connected subset of $\mathbb{R}^n$.

\begin{lem}\label{prop:20160420}
Let $(N, (A_i)_{i\in N}, (u_i)_{i\in N})$ be a finite strategic game. Then the noncooperative payoff region $S_{\nc}$ is path-connected and compact.
\end{lem}

\begin{proof}
Suppose that $A_i = \{a_i^1, \dots, a_i^{m_i}\}$ for every $i\in N$. Then the set $\Delta(A_i)$ can be identified with the standard $(m_i - 1)$-simplex $\Delta^{m_i - 1}$. For each $j\in N$, define $U_j\colon \prod_{i\in N} \mathbb{R}^{m_i}\to \mathbb{R}$ by
\[
U_j(x_1, \dots, x_n) = \sum_{(a_1^{r_1}, \dots, a_n^{r_n})\in \prod_{i\in N} A_i} x_1^{r_1} \cdots x_n^{r_n} u_j(a_1^{r_1}, \dots, a_n^{r_n}),
\]
where $x_i = (x_i^1, \dots, x_i^{m_i})\in \mathbb{R}^{m_i}$ for every $i\in N$. Define the vector-valued function $U\colon \prod_{i\in N} \mathbb{R}^{m_i}\to \mathbb{R}^n$ by
\[
U(x_1, \dots, x_n) = (U_1(x_1, \dots, x_n), \dots, U_n(x_1, \dots, x_n))
\]
for all $(x_1, \dots, x_n)\in \prod_{i\in N} \mathbb{R}^{m_i}$. Then $U$ is a continuous function with respect to the Euclidean metric topologies on $\prod_{i\in N} \mathbb{R}^{m_i}$ and $\mathbb{R}^n$. Obviously, the set $\prod_{i\in N}\Delta^{m_i - 1}$ is a path-connected, compact subset of $\prod_{i\in N} \mathbb{R}^{m_i}$. We see that the noncooperative payoff region $S_{\nc}$ is path-connected and compact, since this region is the image of $\prod_{i\in N}\Delta^{m_i - 1}$ under the continuous function $U$.
\end{proof}

We now propose a simple approach to describe the shape of a noncooperative payoff region by its extreme points and supporting hyperplanes.

\subsection*{Extreme Points}

Extreme points of convex sets play an important role in convex analysis, but here we extend the concept of extreme points to non-convex sets. This simple extension turns out to be very useful in some applications such as describing the shape of a noncooperative payoff region,
and it provides an efficient approach for analysis of Pareto efficiency and social efficiency in game theory; see Section~\ref{sec:App}.

Under the notion of extreme points of a convex set, any corner point of a closed convex polygon is an extreme point. When a non-convex set is considered, a minor modification is required in order to retain this property, as described below.

\begin{Def}\label{def:20170307}
Let $S$ be a nonempty subset of $\mathbb{R}^n$.
\footnote{We do not assume the subset $S$ of $\mathbb{R}^n$ to be convex in this definition.}
A point $x\in S$ is called an \emph{extreme point} of $S$ if
\begin{equation}\label{eq:20170306}
x\in \{\, \theta y + (1-\theta)z \mid \theta\in (0, 1) \,\}\subset S
\end{equation}
with $y$, $z\in S$ implies that $x = y = z$. Let $\ext(S)$ denote the set of all extreme points of $S$.
\end{Def}
This means that an extreme point of $S$ is a point in $S$ which is not an interior point of any line segment contained in $S$. Because the subset $S$ in Definition~\ref{def:20170307} may or may not be convex, it is not sufficient to check whether the point $x$ can lie in the interior of the line segment joining two distinct points of $S$; it must also be checked that the line segment is contained in $S$, that is, the condition $\{\, \theta y + (1-\theta)z \mid \theta\in (0, 1) \,\}\subset S$ in~\eqref{eq:20170306} is required.
\footnote{As we will see in Example~\ref{ex:20170310}, the corner point $(2, 2)$ of the non-convex polygon is, by Definition~\ref{def:20170307}, an extreme point (see Figure~\ref{fig:1116}), although the point $(2, 2)$ can be written as a convex combination of two distinct points in the polygon.}
A closed set may have no extreme points (e.g., a closed half-plane), a finite number of extreme points (e.g., a closed polygon), or an infinite number of extreme points (e.g., a closed disk).

\begin{Rem}\label{rem:20160128}
An interior point of a subset $S$ of $\mathbb{R}^n$ cannot be an extreme point of $S$. In other words, all extreme points must be boundary points.
\end{Rem}

The next lemma is the fundamental result concerning the existence of extreme points of a nonempty compact set; see \citet[p.~21]{tIch:gtea}.

\begin{lem}\label{prop:20160424}
Every nonempty compact subset of $\mathbb{R}^n$ has extreme points.
\end{lem}

A basic property of convex sets in $\mathbb{R}^n$ states that any compact convex subset $K\subset \mathbb{R}^n$ is the convex hull of its extreme points, that is, $K = \conv( \ext(K) )$.
In addition, for every compact subset $S\subset K$ such that $\conv(S) = K$, we have $\ext(K)\subseteq S$.
These are obvious consequences of the Krein--Milman theorem and Milman's theorem in a finite-dimensional case (see, e.g., \citet[pp.~75--76]{wRud:fa}).
Thus, for a given finite strategic game, the equations in~\eqref{eq:20170309} lead to the relations:
\[
S_{\co} = \conv( \ext(S_{\co}) )
\quad \text{and} \quad
\ext(S_{\co})\subseteq S_{\pu}\subseteq S_{\nc}.
\]
The primary purpose of this paper is to explore the relevant properties of a noncooperative payoff region.
The following theorem states the relations between the extreme points of $S_{\nc}$ and other payoff points, which will play a central role in numerous applications.

\begin{thm}\label{prop:20151111}
Let $(N, (A_i)_{i\in N}, (u_i)_{i\in N})$ be a finite strategic game. Then
\[
\varnothing\neq \ext(S_{\co})\subseteq \ext(S_{\nc})\subseteq S_{\pu}.
\]
\end{thm}

\begin{proof}
By Lemma~\ref{prop:20160424}, we have $\ext(S_{\co})\neq \varnothing$. To show that $\ext(S_{\co})\subseteq \ext(S_{\nc})$, let $v\in \ext(S_{\co})$.
Then $v\in S_{\nc}$ because it is known that $\ext(S_{\co})\subset S_{\nc}$. If $v\notin \ext(S_{\nc})$, then there exist distinct points $x$, $y\in S_{\nc}$ such that $v = \lambda x + (1 - \lambda)y$ for some $\lambda\in (0,1)$. Since the convex set $S_{\co}$ contains $S_{\nc}$, we obtain that $v\notin \ext(S_{\co})$, a contradiction.

Next we show that $\ext(S_{\nc})\subseteq S_{\pu}$. Let $v\in \ext(S_{\nc})$ and suppose that $v = u(\sigma)$ for some mixed-strategy profile $\sigma\in \prod_{i\in N}\Delta(A_i)$. Assuming $A_i = (a_i^1, \dots, a_i^{m_i})$ for each $i\in N$, the following decomposition holds: for some $j\in N$,
\[
v = \alpha_j^1 u(a_j^1, \sigma_{-j}) + (1 - \alpha_j^1)\left( \sum_{t=2}^{m_j} \frac{\alpha_j^t}{1 - \alpha_j^1} u(a_j^t, \sigma_{-j}) \right),
\]
where $\alpha_j^t = \sigma_j(a_j^t)$ for $t = 1, \dots, m_j$. Since $v\in \ext(S_{\nc})$ and the closed line segment joining $u(a_j^1, \sigma_{-j})$ and $\sum_{t=2}^{m_j} \frac{\alpha_j^t}{1 - \alpha_j^1} u(a_j^t, \sigma_{-j})$ is contained in $S_{\nc}$, we have
\[
v = u(a_j^1, \sigma_{-j}) = \sum_{t=2}^{m_j} \frac{\alpha_j^t}{1 - \alpha_j^1} u(a_j^t, \sigma_{-j}).
\]
If $(a_j^1, \sigma_{-j})\in \prod_{i\in N}A_i$, then $v\in S_{\pu}$ and we are done. Otherwise, we proceed by decomposing the mixed-strategy profile $(a_j^1, \sigma_{-j})$ as above. Continue this process. Since the set $N$ is finite, eventually we will reach the conclusion that $v = u(a)$ for some $a\in \prod_{i\in N}A_i$.
\end{proof}

It is well known that the case of a two-dimensional noncooperative payoff region with a curved boundary is very common, as we will see in Section~\ref{sec:examples}. Some curved part of the boundary would look like a parabola, which is the envelope of a family of line segments generated by players' mixed strategies. This indicates that the noncooperative payoff region is always outside such a parabola. Thus any subregion containing a relative neighborhood of a boundary point of the noncooperative payoff region cannot be strictly convex.
This assertion can be generalized to higher dimensions, and it can be easily proved using Theorem~\ref{prop:20151111}, as stated below.

\begin{Rem}\label{rem:20160430}
First note that in $\mathbb{R}^n$ boundary points and extreme points coincide if a set is strictly convex. For an $n$-player finite strategic game, Theorem~\ref{prop:20151111} points out that the noncooperative payoff region $S_{\nc}$ has only a finite number of extreme points. So any subregion must be non-strictly convex, provided that the subregion contains a relative neighborhood of a boundary point of $S_{\nc}$.
\end{Rem}

\bigskip

Every closed convex set in $\mathbb{R}^n$ can be represented by closed half-spaces.
Accordingly, for any $n$-player finite strategic game, the cooperative payoff region $S_{\co}$ is the intersection of all closed half-spaces containing it.
Similarly,
the convex hull of the noncooperative payoff region, $\conv(S_{\nc})$,
can be expressed as the intersection of all closed half-spaces containing $S_{\nc}$ (see, e.g., \citet[p.~99]{tRoc:ca}).
The fact that $S_{\co}$ equals $\conv(S_{\nc})$ gives us a motivation to characterize $S_{\nc}$ in terms of the supporting hyperplanes to $S_{\co}$.

\subsection*{Supporting Hyperplanes}

For any given $c\in \mathbb{R}^n\setminus \{ 0 \}$ and $\alpha\in \mathbb{R}$, the set
\[
H = \{\, x\in \mathbb{R}^n \mid c\cdot x = \alpha \,\}
\]
is called a \emph{hyperplane} in $\mathbb{R}^n$. The sets
\[
H^{-} = \{\, x\in \mathbb{R}^n \mid c\cdot x\leq \alpha \,\}
\quad \text{and} \quad
H^{+} = \{\, x\in \mathbb{R}^n \mid c\cdot x\geq \alpha \,\}
\]
are called the \emph{closed half-spaces} determined by $H$.

\begin{Def}
Let $S$ be a nonempty closed subset of $\mathbb{R}^n$. A hyperplane $H$ is said to be a \emph{supporting hyperplane} to $S$ if $S\cap H\neq \varnothing$ and $S$ is contained in one of the two closed half-spaces determined by $H$.
\end{Def}

\begin{lem}\label{prop:20160421}
Let $S_1$ and $S_2$ be nonempty closed subsets of $\mathbb{R}^n$ satisfying $\conv(S_1) = \conv(S_2)$. Then the sets $S_1$ and $S_2$ have the same supporting hyperplanes.
\end{lem}

\begin{proof}
Let $H = \{\, x\in \mathbb{R}^n \mid c\cdot x = \alpha \,\}$ be a supporting hyperplane to $S_1$, and suppose that $S_1\subseteq H^+$. Since $H^+$ is a convex set and $\conv(S_1)$ is the smallest convex set containing $S_1$, we have $S_1\subseteq \conv(S_1)\subseteq H^+$. Thus $S_2\subseteq \conv(S_2)\subseteq H^+$.

Let $\bar{x}\in S_1\cap H$. This implies that $\bar{x}\in \conv(S_2)$ and $c\cdot \bar{x} = \alpha$.
So $\bar{x}$ can be written as a convex combination $\bar{x} = \sum_{t=1}^k \lambda_t y_t$ of some points $y_1, \dots, y_k\in S_2$
where $\lambda_1 + \dots + \lambda_k = 1$ and $\lambda_t> 0$ for all $t$, and then we get $\sum_{t=1}^k \lambda_t c\cdot y_t = \alpha$.
Since each $y_t\in S_2$ and $S_2\subseteq H^+$, the relation $c\cdot y_t\geq \alpha$ holds for each $t$.
It follows that $c\cdot y_t = \alpha$ for every $t$, and then $y_t\in S_2\cap H$ for every $t$.
Therefore, $H$ is also a supporting hyperplane to $S_2$.
\end{proof}

That the noncooperative and cooperative payoff regions of a finite strategic game have the same supporting hyperplanes is a direct result of Lemma~\ref{prop:20160421}.

\begin{thm}\label{prop:20160425}
Let $(N, (A_i)_{i\in N}, (u_i)_{i\in N})$ be a finite strategic game. Then a hyperplane $H$ in $\mathbb{R}^n$ is a supporting hyperplane to $S_{\nc}$ if and only if it is a supporting hyperplane to $S_{\co}$.
\end{thm}

An important property of convex sets is that every supporting hyperplane to a nonempty compact convex subset $K$ of $\mathbb{R}^n$ contains at least one extreme point of $K$ (see, e.g., \citet[p.~303]{jMoo:mmet1}). Hence, every supporting hyperplane to a convex polytope $S_{\co}$ must have an extreme point of $S_{\co}$. Likewise, although a noncooperative payoff region $S_{\nc}$ may not be convex, every supporting hyperplane to $S_{\nc}$ also contains at least one extreme point of $S_{\nc}$, which can be achieved by using a pure-strategy profile, in fact.

\begin{cor}\label{cor:20170314}
Let $(N, (A_i)_{i\in N}, (u_i)_{i\in N})$ be a finite strategic game. Then every supporting hyperplane to $S_{\nc}$ contains at least one extreme point of $S_{\nc}$.
\end{cor}

\begin{proof}
Let $H$ be a supporting hyperplane to $S_{\nc}$. Then $H$ is also a supporting hyperplane to $S_{\co}$ by Theorem~\ref{prop:20160425}, and thus it contains an extreme point of $S_{\co}$. Since the relation $\ext(S_{\co})\subseteq \ext(S_{\nc})$ holds by Theorem~\ref{prop:20151111}, we can conclude that the hyperplane $H$ must contain an extreme point of $S_{\nc}$.
\end{proof}


\section{Examples of Extreme Points of Payoff Regions}
\label{sec:examples}


We now give some examples of extreme points of payoff regions.
As shown below, these examples reveal that all ``convex corners'' of a non-convex set are extreme points,
and thus the definition of an extreme point of a non-convex set in this paper is intuitively reasonable.

It is notable that an extreme point of a payoff region is achieved as a payoff profile on which players, instead of selecting pure strategies, may choose mixed strategies. In this section, we further clarify the relations among extreme points, \emph{pure-payoff profiles}, and \emph{non-mixable payoff profiles}.
Here a pure-payoff profile refers to a payoff profile generated by choosing pure strategies;
a non-mixable payoff profile refers to a payoff profile that can be generated only by a pure-strategy profile.

\begin{example}
Consider the following two-player game.

\begin{table}[!h]
\centering
\renewcommand{\arraystretch}{1.2}
\begin{tabular}{r|c|c|}
\multicolumn{1}{c}{} & \multicolumn{1}{c}{$a_{21}$} & \multicolumn{1}{c}{$a_{22}$}\\ \cline{2-3}
      $a_{11}$ & $0$, $1$ & $5$, $2$ \\ \cline{2-3}
      $a_{12}$ & $2$, $5$ & $1$, $0$ \\ \cline{2-3}
\end{tabular}
\end{table}

\noindent
We can see that $\ext(S_{\co}) = \ext(S_{\nc}) = S_{\pu}$. The noncooperative payoff region of this game is shown in Figure~\ref{fig:1118}.

\begin{figure}[!htb]
  \centering
  \includegraphics[width=0.63\textwidth]{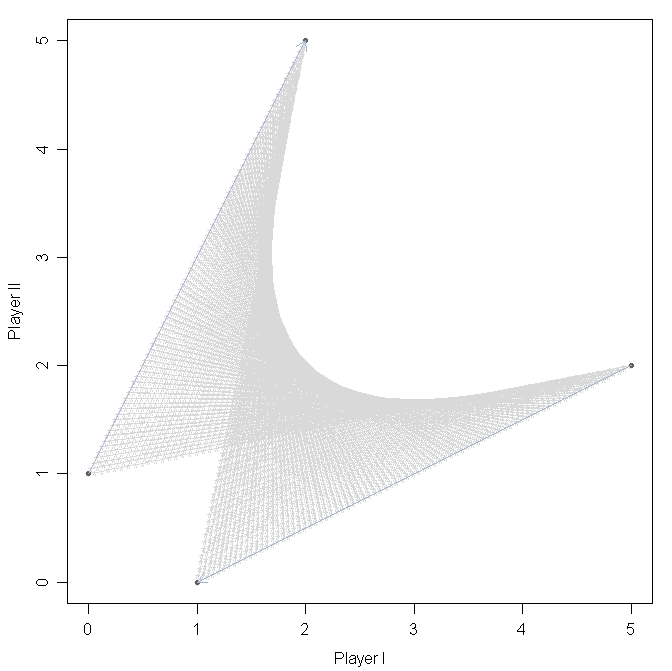}
  \caption{$\ext(S_{\co}) = \ext(S_{\nc}) = S_{\pu}$.}\label{fig:1118}
\end{figure}
\end{example}

\begin{example}\label{ex:20170310}
Consider the following two-player game.

\begin{table}[!h]
\centering
\renewcommand{\arraystretch}{1.2}
\begin{tabular}{r|c|c|c|}
\multicolumn{1}{c}{} & \multicolumn{1}{c}{$a_{21}$} & \multicolumn{1}{c}{$a_{22}$} & \multicolumn{1}{c}{$a_{23}$}\\ \cline{2-4}
      $a_{11}$ & $0$, $2$ & $0$, $1$ & $0$, $0$ \\ \cline{2-4}
      $a_{12}$ & $3$, $0$ & $3$, $2$ & $2$, $2$ \\ \cline{2-4}
\end{tabular}
\end{table}

\noindent
It is easy to see that
$\ext(S_{\co}) = \{ (0,0), (3,0), (3,2), (0,2) \}$,
$\ext(S_{\nc}) = \ext(S_{\co})\cup \{ (2,2) \}$, and $S_{\pu} = \ext(S_{\nc})\cup \{ (0,1) \}$. Hence $\ext(S_{\co})\varsubsetneq \ext(S_{\nc})\varsubsetneq S_{\pu}$. The noncooperative payoff region of this game is shown in Figure~\ref{fig:1116}.

\begin{figure}[!htb]
  \centering
  \includegraphics[width=0.63\textwidth]{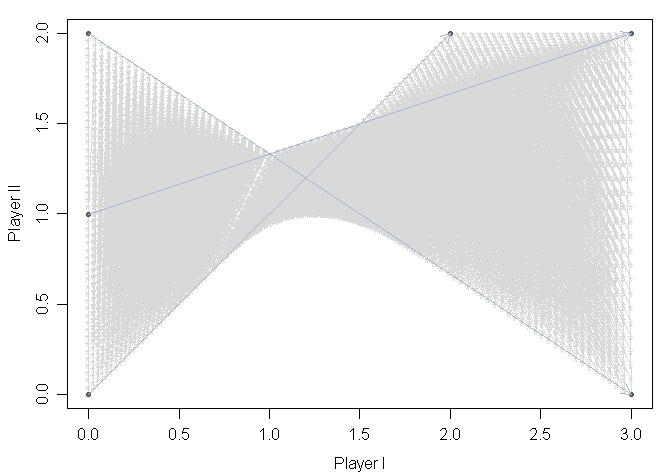}
  \caption{$\ext(S_{\co})\varsubsetneq \ext(S_{\nc})\varsubsetneq S_{\pu}$.}\label{fig:1116}
\end{figure}
\end{example}

For a finite strategic game $(N, (A_i)_{i\in N}, (u_i)_{i\in N})$, a payoff profile $v\in \mathbb{R}^n$ is called \emph{mixable} if it can be generated by a non-degenerate mixed-strategy profile, that is,
\[
v\in \left\{\, u(\sigma) \biggm| \sigma\in \prod_{i\in N}\Delta(A_i)\setminus \prod_{i\in N}A_i \,\right\}.
\]
A payoff profile is said to be non-mixable if it can be generated only by a pure-strategy profile, and it is certainly a pure-payoff profile. However, a pure-payoff profile is not necessarily non-mixable. For example, the pure-payoff pair $(0, 1)$ in Figure~\ref{fig:1116} is mixable. On the other hand, all the extreme points in Figures~\ref{fig:1118} and~\ref{fig:1116} are non-mixable. Nevertheless, as the following example shows, an extreme point of a noncooperative payoff region and a non-mixable payoff profile of a finite strategic game are not the same thing.

\begin{example}
Consider the following two-player game $G_1$.

\begin{table}[!h]
\centering
\renewcommand{\arraystretch}{1.2}
\begin{tabular}{r|c|c|}
\multicolumn{1}{c}{} & \multicolumn{1}{c}{$a_{21}$} & \multicolumn{1}{c}{$a_{22}$}\\ \cline{2-3}
      $a_{11}$ & $\phantom{-}0$, $\phantom{-}0$ & $0$, $\phantom{-}0$ \\ \cline{2-3}
      $a_{12}$ &           $-1$,           $-1$ & $1$,           $-1$ \\ \cline{2-3}
\end{tabular}
\end{table}

\noindent
Obviously, the payoff pair $(0,0)$ of $G_1$ is mixable, and it is an extreme point of the noncooperative payoff region of $G_1$.
\footnote{The noncooperative payoff region of $G_1$ is a triangle with vertices $(0,0)$, $(-1,-1)$, and $(1,-1)$.}
In addition, a non-mixable payoff profile may not be an extreme point of a noncooperative payoff region, as shown in the following two-player game $G_2$.

\begin{table}[!h]
\centering
\renewcommand{\arraystretch}{1.2}
\begin{tabular}{r|c|c|}
\multicolumn{1}{c}{} & \multicolumn{1}{c}{$a_{21}$} & \multicolumn{1}{c}{$a_{22}$}\\ \cline{2-3}
      $a_{11}$ &           $-1$, $0$ & $0$,           $-1$ \\ \cline{2-3}
      $a_{12}$ & $\phantom{-}0$, $0$ & $1$, $\phantom{-}0$ \\ \cline{2-3}
\end{tabular}
\end{table}

\noindent
It is easy to check that the payoff pair $(0,0)$ of $G_2$ is non-mixable, and it is not an extreme point of the noncooperative payoff region of $G_2$, which is illustrated in Figure~\ref{fig:0119}.

\begin{figure}[!htb]
  \centering
  \includegraphics[width=0.63\textwidth]{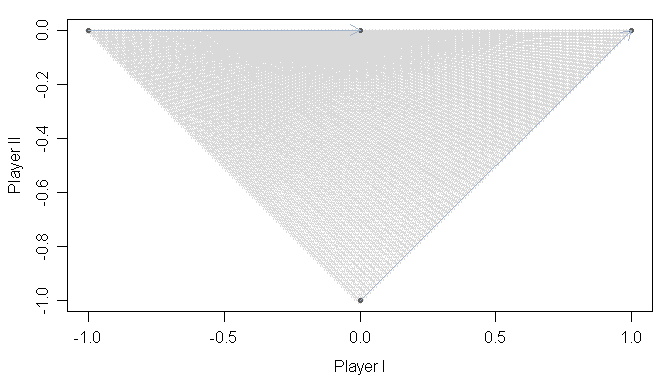}
  \caption{The non-mixable payoff pair $(0,0)$ is not an extreme point of $S_{\nc}$.}\label{fig:0119}
\end{figure}
\end{example}

\bigskip

In convex analysis, a basic result, called the supporting hyperplane theorem, states that if $x$ is a boundary point of a closed convex subset $C$ of $\mathbb{R}^n$, then there exists at least one supporting hyperplane to $C$ at $x$ (see, e.g., \citet[p.~51]{sBoy-lVan:co}). By applying Remark~\ref{rem:20160128}, we obtain that for a cooperative payoff region $S_{\co}$, there exists a supporting hyperplane passing through any given extreme point of $S_{\co}$. In contrast to Corollary~\ref{cor:20170314}, this need not be the case for a noncooperative payoff region as indicated in the example below.

\begin{example}
Consider the following two-player game.

\begin{table}[!h]
\centering
\renewcommand{\arraystretch}{1.2}
    \begin{tabular}{r|c|c|c|}
      \multicolumn{1}{c}{} & \multicolumn{1}{c}{$a_{21}$} & \multicolumn{1}{c}{$a_{22}$}
                           & \multicolumn{1}{c}{$a_{23}$}\\ \cline{2-4}
      $a_{11}$ & $4$, $4$ & $0$, $0$ & $0$, $0$ \\  \cline{2-4}
      $a_{12}$ & $0$, $0$ & $8$, $2$ & $0$, $0$ \\  \cline{2-4}
      $a_{13}$ & $0$, $0$ & $0$, $0$ & $2$, $8$ \\  \cline{2-4}
    \end{tabular}
\end{table}

\noindent
We can see that the payoff pair $(4,4)$ is an extreme point of $S_{\nc}$, and it is an interior point of $S_{\co}$. Hence, applying Theorem~\ref{prop:20160425}, there exists no supporting hyperplane to $S_{\nc}$ passing through the point $(4,4)$, which can be seen in Figure~\ref{fig:0128}.

\begin{figure}[!htb]
  \centering
  \includegraphics[width=0.63\textwidth]{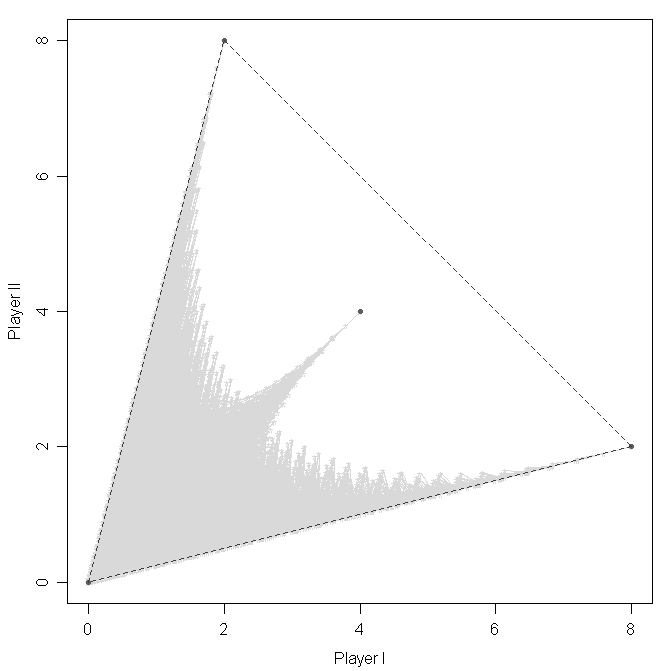}
  \caption{The extreme point $(4,4)$ of $S_{\nc}$ is an interior point of $S_{\co}$.}\label{fig:0128}
\end{figure}
\end{example}


\section{Applications}
\label{sec:App}


The study of extreme points of a non-convex set not only helps us to describe the shape of a noncooperative payoff region,
but enables us to efficiently prove some results about Pareto efficiency and social efficiency in the theory of games.

\begin{Def}
A point $w\in \mathbb{R}^n$ \emph{Pareto dominates} a point $v\in \mathbb{R}^n$ if $w\neq v$ and $w_i\geq v_i$ for $i = 1, \dots, n$. Let $S$ be a nonempty subset of $\mathbb{R}^n$. The \emph{Pareto frontier} of $S$ is defined as
\[
P(S) = \{\, v\in S\mid \{\, w\in S\mid \text{$w$ Pareto dominates $v$} \,\} = \varnothing \,\}.
\]
\end{Def}

In other words, a point $v$ on the Pareto frontier of $S$ means that $v\in S$ and there is no \emph{other} point $w\in S$ such that $w_i\geq v_i$ for all $i$.
For a finite strategic game, we will say more about a payoff pair on the Pareto frontier of the noncooperative payoff region in what follows.

\subsection*{On the Payoff Pairs with Horizontal/Vertical Tangents}

According to \citet{jMay:etg} and \citet{gVic-cCan:dess}, an evolutionarily stable strategy is defined as a strategy adopted by all individuals with the property that no mutant strategy can invade under the influence of natural selection if the population share of mutants is smaller than an invasion barrier.
The \emph{indirect evolutionary approach} is a branch of evolutionary game theory in which individuals are characterized by preferences rather than pre-programmed strategies.
It has received significant attention because, underlying this evolutionary approach with observable preferences,
efficiency is a necessary condition for outcomes to be stable (see, e.g., \citet{wGut-mYaa:erbssg}, \citet{lSam:iep}, and \citet{eDek-jEly-oYil:ep}).

When we check whether a Pareto-efficient strategy pair is stable in an indirect evolutionary two-population model, mutant strategy pairs may have the following property:
one mutant type in a mutant pair receives more than the incumbents in one population, and the other mutant type receives less than the incumbents in the other population.
Even though any such mutant pair would always be driven out if its population share is lower than its corresponding barrier,
the existence of a uniform invasion barrier is still ambiguous, especially as the strategies of the mutants and the incumbents get closer and closer to each other.

This argument prompted us to pay attention to the sequence of the mutant pairs in which the two fitness gaps with the incumbents in the two populations respectively will be closed in opposite directions. In fact, the situation that the corresponding invasion barrier could become arbitrarily small will imply that the trailing gap, caused by mutants trailing behind the incumbents in one population, actually closes faster than the leading gap, caused by leading mutants in the other population. But this contradicts the shape of a noncooperative payoff region. Thus no such mutant strategy pairs can exist.

For a more precise statement of this result, let $S_{\nc}\subset \mathbb{R}^2$ be a noncooperative payoff region of a two-player finite strategic game,
and let $(v_1^*, v_2^*)$ be a payoff pair in $S_{\nc}$. Suppose that $\{(v_1^t, v_2^t)\}$ is a sequence of payoff pairs in $S_{\nc}$ converging to $(v_1^*, v_2^*)$ with $v_1^t - v_1^*> 0$, defining a leading gap, and $v_2^* - v_2^t> 0$, defining a trailing gap, for each $t$. If the curve connecting the sequence has a horizontal tangent line at $(v_1^*, v_2^*)$, then the payoff pair $(v_1^*, v_2^*)$ will not lie on the Pareto frontier of $S_{\nc}$, as in Figure~\ref{fig:20161112}. The intuitive reason for this is clear:
if $(v_1^*, v_2^*)$ is Pareto efficient relative to $S_{\nc}$ and the sequence satisfying the above conditions exists,
then there would be a strictly convex subregion of $S_{\nc}$ containing this sequence, and it would contradict Remark~\ref{rem:20160430}.
We will give a rigorous proof below using the properties of noncooperative payoff regions discussed in Section~\ref{sec:PayReg}.

\begin{figure}[!htb]
  \centering
  \includegraphics[width=0.63\textwidth]{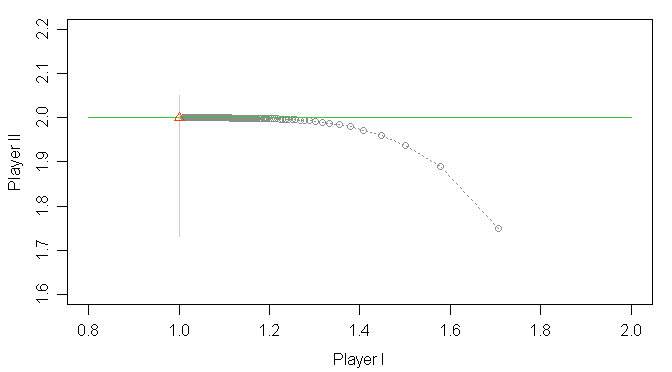}
  \caption{The payoff pair $(1,2)$ will not lie on the Pareto frontier of $S_{\nc}$.}\label{fig:20161112}
\end{figure}

The actual noncooperative payoff region is often too involved, even though there are just two players in a game.
Here is another way of determining whether a payoff pair does not lie on the Pareto frontier of a noncooperative payoff region.

\begin{thm}\label{prop:20161109}
Let $(\{1, 2\}, A_1, A_2, u_1, u_2)$ be a two-player finite strategic game. Suppose that $\{(v_1^t, v_2^t)\}$ is a sequence in $S_{\nc}$ which converges to $(v_1^*, v_2^*)\in S_{\nc}$ and satisfies $v_1^t> v_1^*$ and $v_2^*> v_2^t$ for all $t$. If the condition
\[
\lim_{t\rightarrow \infty}\frac{v_2^t - v_2^*}{v_1^t - v_1^*} = 0
\]
holds, then $(v_1^*, v_2^*)\notin P(S_{\nc})$.
\end{thm}

\begin{proof}
Suppose that $(v_1^*, v_2^*)\in P(S_{\nc})$. This means that for each $(x_1, x_2)\in S_{\nc}$, we have $v_2^*> x_2$ if $x_1> v_1^*$, and $v_2^*\geq x_2$ if $x_1 = v_1^*$. Since $\ext(S_{\nc})$ is a finite set by Theorem~\ref{prop:20151111}, we can pick $v_1^k$ such that
\begin{equation}\label{eq:0208}
\{\, (x_1, x_2)\in \ext(S_{\nc})\mid v_1^*< x_1< v_1^k \,\} = \varnothing.
\end{equation}
By the compactness of $S_{\nc}$ (see Lemma~\ref{prop:20160420}), we can choose $\bar{v}_2^k = \max\{\, \beta\in \mathbb{R}\mid (v_1^k, \beta)\in S_{\nc} \,\}$, and then $v_2^k\leq \bar{v}_2^k< \bar{v}_2^*$. Consider the compact set
\[
A = \{\, (x_1, x_2)\in S_{\nc} \mid v_1^*\leq x_1\leq v_1^k \,\},
\]
and define a linear function $f\colon A\to \mathbb{R}$ by
\[
f(x_1, x_2) = (v_2^* - \bar{v}_2^k)x_1 + (v_1^k - v_1^*)x_2.
\]
Under our assumptions, $f$ achieves its maximum on $A$, say $m$. Furthermore, it is true that $(x_1, x_2)\in f^{-1}(m)$ implies $v_1^*< x_1< v_1^k$, which is obvious from the assumption that $\lim_{t\rightarrow \infty}\frac{v_2^t - v_2^*}{v_1^t - v_1^*} = 0$, and from the fact that the level curves of $f$ are parallel straight line segments with slope $\frac{\bar{v}_2^k - v_2^*}{v_1^k - v_1^*}$.

We claim $\ext(f^{-1}(m))\subseteq \ext(S_{\nc})$. To see this, assume that $(\bar{x}_1, \bar{x}_2)\in f^{-1}(m)$ and $(\bar{x}_1, \bar{x}_2)\notin \ext(S_{\nc})$. Then $v_1^*< \bar{x}_1< v_1^k$, and we can choose distinct points $(y_1, y_2)$, $(z_1, z_2)\in A$ such that
\[
(\bar{x}_1, \bar{x}_2)\in \{\, \theta (y_1, y_2) + (1-\theta)(z_1, z_2) \mid \theta\in (0, 1) \,\}\subset A.
\]
The fact that $(\bar{x}_1, \bar{x}_2)\in f^{-1}(m)$ implies that $(y_1, y_2)$, $(z_1, z_2)\in f^{-1}(m)$, and thus $\theta (y_1, y_2) + (1-\theta)(z_1, z_2)\in f^{-1}(m)$ for every $\theta\in (0, 1)$. This means that $(\bar{x}_1, \bar{x}_2)\notin \ext(f^{-1}(m))$. Moreover, since $f^{-1}(m)$ is compact, the set $\ext(f^{-1}(m))$ is nonempty by Lemma~\ref{prop:20160424}. Therefore we get the relations
\[
\varnothing\neq \ext(f^{-1}(m))\subseteq \ext(S_{\nc})\cap \{\, (x_1, x_2)\in S_{\nc}\mid v_1^*< x_1< v_1^k \,\},
\]
contradicting~(\ref{eq:0208}).
\end{proof}

\begin{Rem}
By mirror symmetry, there should be a result corresponding to Theorem~\ref{prop:20161109}.
That is,
if a sequence $\{(v_1^t, v_2^t)\}$ in $S_{\nc}$ converges to $(v_1^*, v_2^*)\in S_{\nc}$ satisfying $v_1^t< v_1^*$ and $v_2^*< v_2^t$ for all $t$,
then the limit point $(v_1^*, v_2^*)$ does not lie on the Pareto frontier of $S_{\nc}$,
provided that the curve connecting the sequence has a vertical tangent line at $(v_1^*, v_2^*)$.
\end{Rem}

\subsection*{On the Rational Payoff Regions}

In some environments, like infinitely repeated games, the coefficients used in a convex combination of pure-payoff profiles are sometimes restricted to rational numbers. For an $n$-player finite strategic game, we denote by $\rcvx{S}_{\co}$ the set of all convex combinations with \emph{rational} coefficients of the pure-payoff profiles. Similarly, we denote by $\rcvx{S}_{\nc}$ the set of all payoff profiles achievable with \emph{rational} mixed-strategy profiles.
\footnote{A mixed strategy $\sigma_i$ for player~$i$ is \emph{rational} if all values of $\sigma_i$ are rational numbers. A mixed-strategy profile $\sigma$ is \emph{rational} if all its mixed strategies are rational.}
Then $\rcvx{S}_{\co}$ is dense in $S_{\co}$, and $\rcvx{S}_{\nc}$ is dense in $S_{\nc}$.

If $\sigma$ is a rational mixed-strategy profile and $u(\sigma)\in P(S_{\nc})$, then it is clear that $u(\sigma)\in P(\rcvx{S}_{\nc})$.
Conversely, if the payoff profile $u(\sigma)$ lies on the Pareto frontier of $\rcvx{S}_{\nc}$, which is dense in $S_{\nc}$,
do we get the result that $u(\sigma)\in P(S_{\nc})$?
In a two-player finite strategic game, the answer is yes and it also holds for a cooperative payoff region, as the following theorem shows.

\begin{thm}
Let $(\{1, 2\}, A_1, A_2, u_1, u_2)$ be a two-player finite strategic game. Then $P(\rcvx{S}_{\nc})\subseteq P(S_{\nc})$ and $P(\rcvx{S}_{\co})\subseteq P(S_{\co})$.
\end{thm}

\begin{proof}
The proofs of these two cases are similar; we only show that $P(\rcvx{S}_{\nc})\subseteq P(S_{\nc})$. Let $v\in \rcvx{S}_{\nc}$ and $v\notin P(S_{\nc})$. Then there exists $w\in S_{\nc}$ such that $w$ Pareto dominates $v$. If $w_1> v_1$ and $w_2> v_2$, then since $\rcvx{S}_{\nc}$ is dense in $S_{\nc}$, we can choose $\bar{x}\in \rcvx{S}_{\nc}$ such that $\bar{x}_1> v_1$ and $\bar{x}_2> v_2$, which means that $v\notin P(\rcvx{S}_{\nc})$.

Otherwise, without loss of generality, we assume that $w_1> v_1$ and $w_2 = v_2$. Since $S_{\nc}$ is a compact set (see Lemma~\ref{prop:20160420}), we can take
\[
\bar{w}_1 = \max \{\, \alpha\in \mathbb{R} \mid (\alpha, v_2)\in S_{\nc} \,\},
\]
and then we have $\bar{w}_1\geq w_1> v_1$. If $(\bar{w}_1, v_2)\in \ext(S_{\nc})$, then, by Theorem~\ref{prop:20151111}, the point $(\bar{w}_1, v_2)$ is a pure-payoff pair, and hence $(\bar{w}_1, v_2)\in \rcvx{S}_{\nc}$. Therefore $v\notin P(\rcvx{S}_{\nc})$.

If $(\bar{w}_1, v_2)\notin \ext(S_{\nc})$, we can choose two distinct points $u$, $u'\in S_{\nc}$ such that
\[
(\bar{w}_1, v_2)\in \{\, \theta u + (1-\theta)u' \mid \theta\in (0, 1) \,\}\subset S_{\nc}.
\]
The way in which $\bar{w}_1$ is chosen leads to $u_2\neq v_2$ and $u'_2\neq v_2$. Without loss of generality, assume that $u_2> v_2$. Since $\bar{w}_1> v_1$, we can always choose $\lambda\in (0, 1)$ such that the point $(x_1, x_2) = \lambda(\bar{w}_1, v_2) + (1 - \lambda)u$ belongs to $S_{\nc}$, and the conditions $x_1> v_1$ and $x_2> v_2$ are satisfied. Again, since $\rcvx{S}_{\nc}$ is dense in $S_{\nc}$, there exists $\bar{x}\in \rcvx{S}_{\nc}$ such that $\bar{x}$ strictly Pareto dominates $v$, that is, $v\notin P(\rcvx{S}_{\nc})$.
\end{proof}

The relations in the above theorem seem intuitively clear, because for every payoff profile, it is either achieved using rational coefficients or arbitrarily close to one with rational coefficients. But the above results cannot be extended to general $n$-player games. Here we provide a counterexample for $n = 3$.

The main idea for the counterexample is caused by the above proof.
Note that in the last case,
we can always choose a point on the line segment containing the non-extreme point and contained in the payoff region such that the rational payoff pair is strictly Pareto dominated. However, this may not be done if the payoff region is generated from a multi-player game rather than from a two-player game.

\begin{example}
Consider the following three-player game, where the payoff function is denoted by $u$. Figure~\ref{fig:1124} shows the noncooperative payoff region $S_{\nc}$ of this game.

\begin{table}[!h]
\centering
\renewcommand{\arraystretch}{1.2}
    \begin{tabular}{r|c|c|}
      \multicolumn{1}{c}{} & \multicolumn{1}{c}{$a_{21}$} & \multicolumn{1}{c}{$a_{22}$} \\ \cline{2-3}
      $a_{11}$ & $\sqrt{2}$, $-1$, $\phantom{-}1$ & $\phantom{-}0$, $\phantom{\sqrt{}}0$, $-1$ \\ \cline{2-3}
      $a_{12}$ & $\phantom{\sqrt{}}0$, $\phantom{-}0$, $-1$ & $-2$, $\sqrt{2}$, $\phantom{-}1$ \\ \cline{2-3}
      \multicolumn{1}{c}{} & \multicolumn{2}{c}{$a_{31}$}
    \end{tabular} \qquad
    \begin{tabular}{r|c|c|}
      \multicolumn{1}{c}{} & \multicolumn{1}{c}{$a_{21}$} & \multicolumn{1}{c}{$a_{22}$}\\ \cline{2-3}
      $a_{11}$ & $\sqrt{2}$, $-1$, $\phantom{-}1$ & $\phantom{-}0$, $\phantom{\sqrt{}}0$, $-1$ \\ \cline{2-3}
      $a_{12}$ & $\phantom{\sqrt{}}0$, $\phantom{-}0$, $-1$ & $-2$, $\sqrt{2}$, $\phantom{-}1$ \\ \cline{2-3}
      \multicolumn{1}{c}{} & \multicolumn{2}{c}{$a_{32}$}
    \end{tabular}
\end{table}

\noindent
The strategy $a_{11}$ is strictly dominant for player~$1$; the strategy $a_{22}$ is strictly dominant for player~$2$; the actions of player~$3$ have no effect on their payoffs. Let $\sigma^*$ be a Nash equilibrium, in which the players choose their strategies independently. Then $\sigma^*_1 = a_{11}$, $\sigma^*_2 = a_{22}$, and $\sigma^*_3$ is one of the possible strategies of the player~$3$. Therefore, $u(\sigma^*) = (0,0,-1)$.
We claim that $u(\sigma^*)\in P(\rcvx{S}_{\nc})$.
Nevertheless, $u(\sigma^*)\notin P(S_{\nc})$ and $u(\sigma^*)\notin P(S_{\co})$, as shown in Figure~\ref{fig:1124}.

To verify the claim, we show that $u(\sigma^*)\in P(\rcvx{S}_{\co})$;
thus, since $u(\sigma^*)\in \rcvx{S}_{\nc}$ and $\rcvx{S}_{\nc}\subseteq \rcvx{S}_{\co}$, we have $u(\sigma^*)\in P(\rcvx{S}_{\nc})$.
Let $\varphi$ be a correlated strategy with the property that $u(\varphi)$ Pareto dominates $u(\sigma^*)$. This implies that
\[
\varphi(a_{11}, a_{21}, a_{31}) + \varphi(a_{11}, a_{21}, a_{32}) = \sqrt{2}[\varphi(a_{12}, a_{22}, a_{31}) + \varphi(a_{12}, a_{22}, a_{32})],
\]
and $\varphi(a_{1j}, a_{2j}, a_{31}) + \varphi(a_{1j}, a_{2j}, a_{32})\neq 0$ for $j = 1, 2$.
Therefore, $u(\varphi)\notin \rcvx{S}_{\co}$, and so we can conclude that $u(\sigma^*)\in P(\rcvx{S}_{\co})$.

\begin{figure}[!htb]
  \centering
  \includegraphics[width=0.73\textwidth]{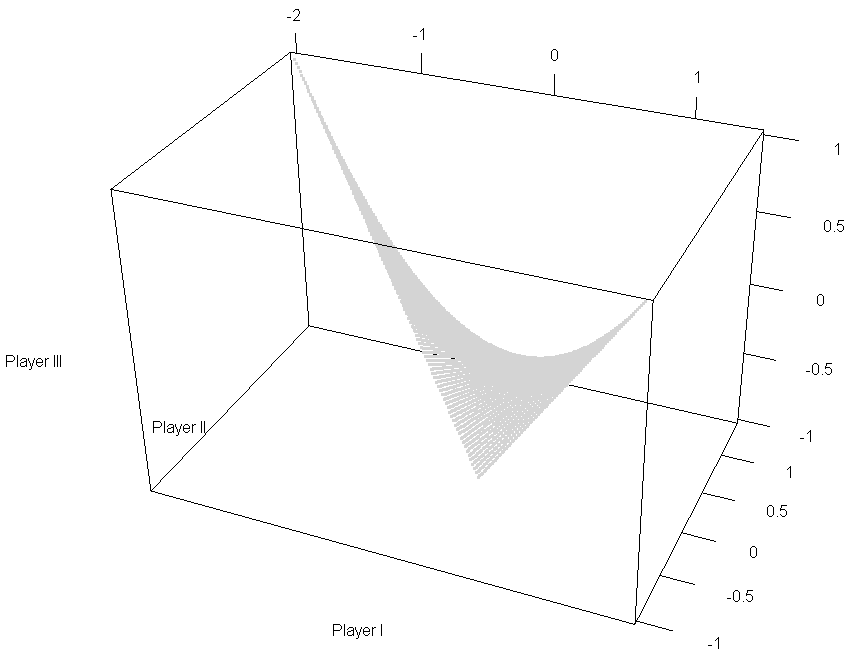}
  \caption{$P(\rcvx{S}_{\co})\nsubseteq P(S_{\co})$; $P(\rcvx{S}_{\nc})\nsubseteq P(S_{\nc})$.}\label{fig:1124}
\end{figure}
\end{example}

\subsection*{On the Social Efficiency}
Pareto efficiency only guarantees that no one can become better off without making someone else worse off. A situation would still be Pareto efficient despite huge disparities among individuals. To reduce inequality in the distribution of utilities, the goal of a decision maker is to select the point on the Pareto frontier that gives the highest social welfare.
However, because there are multiple conceptions of inequality, there is no simple criterion for choosing among those diverse social welfare functions.

Here, for a finite strategic game $(N, (A_i)_{i\in N}, (u_i)_{i\in N})$, we consider a \emph{weighted utilitarian social welfare function}
$W\colon S_{\nc}\to \mathbb{R}$ defined by
\[
W(v) = \sum_{i\in N}\alpha_i v_i,
\]
which is one of the most widely used methods for aggregation of individual utilities. It can be shown that the maximum value of the weighted utilitarian social welfare function can be achieved under a pure-strategy profile, although the shape of a noncooperative payoff region is unclear. Before explaining this, let us first review the importance of extreme points in convex analysis.

Extreme points play a crucial role in solving convex optimization problems as described below.
A continuous convex function on a compact convex subset of $\mathbb{R}^n$ will always attain its maximum at an extreme point of the subset (see, e.g., \citet[p.~298]{cAli-kBor:ida}).
This is known as the Bauer Maximum Principle. It should be emphasized that this result does not mean that all maximizers are extreme points.

Now let the social welfare function $W$ be extended to the domain $S_{\co}$, a compact convex subset of $\mathbb{R}^n$.
Then there exists $v_0\in \ext(S_{\co})$ such that $W(v_0) = \max \{\, W(v)\mid v\in S_{\co} \,\}$.
By Theorem~\ref{prop:20151111}, we know that $\ext(S_{\co})\subseteq \ext(S_{\nc})\subseteq S_{\pu}$.
Therefore, the relations $S_{\pu}\subseteq S_{\nc}\subseteq S_{\co}$ imply that
\[
W(v_0) = \max \{\, W(v)\mid v\in S_{\nc} \,\} = \max \{\, W(v)\mid v\in S_{\pu} \,\}.
\]
Thus, for the weighted utilitarian social welfare function defined on a noncooperative payoff region, its maximum value can be obtained by means of a pure-strategy profile, no matter what the shape of this payoff region is.


\section{Conclusion}


For every $n$-player finite strategic game, the cooperative payoff region is a convex polytope in $\mathbb{R}^n$, and the noncooperative payoff region is a compact connected subset of this polytope. A corner of the cooperative payoff region is a convex corner of the noncooperative payoff region; but not vice versa. This indicates that the two payoff regions have the same supporting hyperplanes.

In addition, although the boundary of a payoff region may have a very complex structure when the players choose their strategies independently, there is one remarkable characteristic common to all noncooperative payoff regions: any subregion must be non-strictly convex if it contains a relative neighborhood of a boundary point of the noncooperative payoff region.

In this paper, we study them in a strict mathematical way. Besides, the properties of extreme points of a noncooperative payoff region also allow us to efficiently prove some new properties, such as the theorems in Section~\ref{sec:App}. This approach not only provides rigorous proofs for obvious assertions, but also helps to clarify some questions to which the intuition gives no good answers. These fully demonstrate that it can be an effective and efficient approach for research purposes.


\bibliographystyle{plainnat}
\bibliography{PayReg}

\begin{thebibliography}{13}
\providecommand{\natexlab}[1]{#1}
\providecommand{\url}[1]{\texttt{#1}}
\expandafter\ifx\csname urlstyle\endcsname\relax
  \providecommand{\doi}[1]{doi: #1}\else
  \providecommand{\doi}{doi: \begingroup \urlstyle{rm}\Url}\fi

\bibitem[Aliprantis and Border(2006)]{cAli-kBor:ida}
Charalambos~D. Aliprantis and Kim Border.
\newblock \emph{Infinite Dimensional Analysis: A Hitchhiker's Guide}.
\newblock Springer-Verlag Berlin Heidelberg, 2006.

\bibitem[Barron(2013)]{eBar:gt}
Emmanuel~N. Barron.
\newblock \emph{Game Theory: An Introduction}.
\newblock John Wiley \& Sons, 2013.

\bibitem[Binmore(2007)]{kBin:pr}
Ken Binmore.
\newblock \emph{Playing for Real: A Text on Game Theory}.
\newblock Oxford University Press, 2007.

\bibitem[Boyd and Vandenberghe(2004)]{sBoy-lVan:co}
Stephen Boyd and Lieven Vandenberghe.
\newblock \emph{Convex Optimization}.
\newblock Cambridge University Press, 2004.

\bibitem[Dekel et~al.(2007)Dekel, Ely, and Yilankaya]{eDek-jEly-oYil:ep}
Eddie Dekel, Jeffrey~C. Ely, and Okan Yilankaya.
\newblock Evolution of preferences.
\newblock \emph{Review of Economic Studies}, 74:\penalty0 685--704, 2007.

\bibitem[G\"{u}th and Yaari(1992)]{wGut-mYaa:erbssg}
Werner G\"{u}th and Menahem~E. Yaari.
\newblock Explaining reciprocal behavior in simple strategic games: An
  evolutionary approach.
\newblock In Ulrich Witt, editor, \emph{Explaining Process and Change:
  Approaches to Evolutionary Economics}, pages 23--34. The University of
  Michigan Press, Ann Arbor, 1992.

\bibitem[Ichiishi(1983)]{tIch:gtea}
Tatsuro Ichiishi.
\newblock \emph{Game Theory for Economic Analysis}.
\newblock Academic Press, 1983.

\bibitem[{Maynard Smith}(1982)]{jMay:etg}
John {Maynard Smith}.
\newblock \emph{Evolution and the Theory of Games}.
\newblock Cambridge University Press, 1982.

\bibitem[Moore(1999)]{jMoo:mmet1}
James~C. Moore.
\newblock \emph{Mathematical Methods for Economic Theory 1}.
\newblock Springer-Verlag Berlin Heidelberg, 1999.

\bibitem[Rockafellar(1970)]{tRoc:ca}
R.~Tyrrell Rockafellar.
\newblock \emph{Convex Analysis}.
\newblock Princeton University Press, 1970.

\bibitem[Rudin(1973)]{wRud:fa}
Walter Rudin.
\newblock \emph{Functional Analysis}.
\newblock McGraw-Hill, 1973.

\bibitem[Samuelson(2001)]{lSam:iep}
Larry Samuelson.
\newblock Introduction to the evolution of preferences.
\newblock \emph{Journal of Economic Theory}, 97:\penalty0 225--230, 2001.

\bibitem[Vickers and Cannings(1987)]{gVic-cCan:dess}
Glenn~T. Vickers and Chris Cannings.
\newblock On the definition of an evolutionarily stable strategy.
\newblock \emph{Journal of Theoretical Biology}, 129:\penalty0 349--353, 1987.

\end{thebibliography}



\end{document}